\documentclass[11pt]{article}
\usepackage{enumerate}
\usepackage{amsfonts}
\usepackage{amssymb}
\usepackage{amsmath}
\usepackage{xspace}
\usepackage{algpseudocode}

\usepackage{graphicx}
\usepackage{multicol}
\usepackage[margin=1in]{geometry}
\usepackage{hyperref}
\usepackage{times}

\usepackage{xspace}
\usepackage{amssymb}
\usepackage{amsmath}
\usepackage{amsthm}
\usepackage{latexsym}
\usepackage{graphicx}
\usepackage{color}
\usepackage{dsfont}

\usepackage{hyperref}

\usepackage{amsmath}

\usepackage{centernot}

\newcommand{\A}{{\cal A}}

\newcommand{\D}{{\cal D}}

\newcommand{\R}{{\cal R}}
\newcommand{\RT}{\texttt{RT}\xspace}

\renewcommand{\P}{{\cal P}}

\newcommand{\G}{{\cal G}}

\newcommand{\Z}{\mathbb{Z}}
\newcommand{\Q}{\mathbb{Q}}

\newcommand{\Nat}{\mbox{I$\!$N}}
\newcommand{\Real}{\mbox{I$\!$R}}

\newcommand{\St}{\mbox{St}}
\newcommand{\nSt}{\mbox{nSt}^r}
\newcommand{\dist}{\mbox{dist}}
\newcommand{\Pot}{\mbox{Pot}}
\newcommand{\PO}{Player~$1$\xspace}
\newcommand{\PT}{Player~$2$\xspace}
\newcommand{\PLi}{Player~$i$\xspace}

\newcommand{\thresh}{\texttt{Th}\xspace}

\newcommand{\MP}{\texttt{MP}\xspace}

\newcommand{\THRESHBUDG}{THRESH-BUD\xspace}
\newcommand{\ATHRESHBUDG}{APPROX-THRESH-BUD\xspace}

\usepackage{amsmath}

\usepackage{centernot}

\newtheorem{lemma}{Lemma}
\newtheorem{theorem}[lemma]{Theorem}

\newtheorem{xmpl}[lemma]{Example}
\newenvironment{example}{\begin{xmpl}\rm}{\end{xmpl}}

\newtheorem{rmark}[lemma]{Remark}
\newenvironment{remark}{\begin{rmark}\rm}{\end{rmark}}

\theoremstyle{definition}
\newtheorem{definition}[lemma]{Definition}
\newenvironment{claim}{{\par \vspace{0.07cm} \noindent{\bf Claim:}}}{\par \vspace{0.07cm}}

\newcommand{\stam}[1]{}
\newcommand{\zug}[1]{\langle #1  \rangle}
\newcommand{\set}[1]{\{ #1 \}}

\newcommand{\Max}{{Max}\xspace}
\newcommand{\Min}{{Min}\xspace}

\usepackage[T1]{fontenc}
\usepackage{authblk}
\title{Infinite-Duration Poorman-Bidding Games\thanks{This paper is a full version of \cite{AHI18}.
This research was supported in part by the Austrian Science Fund (FWF) under grants S11402-N23 (RiSE/SHiNE), Z211-N23 (Wittgenstein Award), and M 2369-N33 (Meitner fellowship).}}
\author{Guy Avni\thanks{guy.avni@ist.ac.at}  \ \ \ Thomas A. Henzinger\thanks{tah@ist.ac.at} \ \ \ Rasmus Ibsen-Jensen\thanks{ribsen@ist.ac.at} \\
IST Austria}
\date{}
\begin{document}
\maketitle
\thispagestyle{empty}

\begin{abstract}
In two-player games on graphs, the players move a token through a graph to produce an infinite path, which determines the winner or payoff of the game. Such games are central in formal verification since they model the interaction between a non-terminating system and its environment. We study {\em bidding games} in which the players bid for the right to move the token. Bidding games with variants of first-price auctions were previously studied: in each round, the players simultaneously submit bids, the higher bidder moves the token, and, in {\em Richman} bidding, pays his bid to the other player whereas in {\em poorman} bidding, pays his bid to the ``bank''. While reachability poorman games have been studied before, we present, for the first time, results on {\em infinite-duration} poorman games. A central quantity in these games is the {\em ratio} between the two players' initial budgets. We show that the favorable properties of reachability poorman games extend to complex qualitative objectives such as parity, similarly to the Richman case: each vertex has a {\em threshold value}, which is a necessary and sufficient ratio with which a player can achieve a goal. Our most interesting results concern quantitative poorman games, namely mean-payoff poorman  games, where we construct optimal strategies depending on the initial ratio. The crux of the proof shows that strongly-connected mean-payoff poorman games are equivalent to {\em biased random-turn games}. The equivalence in itself is interesting, because it does not hold for reachability poorman games and it is richer than the equivalence with uniform random-turn games that Richman bidding exhibit. We also solve the complexity problems that arise in poorman games. 
\end{abstract}

\section{Introduction}
Two-player infinite-duration games on graphs are a central class of games in formal verification \cite{AG11} and have deep connections to foundations of logic \cite{Rab69}. They are used to model the interaction between a system and its environment, and the problem of synthesizing a correct system then reduces to finding a winning strategy in a graph game \cite{PR89}. Theoretically, they have been widely studied. For example, the problem of deciding the winner in a parity game is a rare problem that is in NP and coNP \cite{Jur98}, not known to be in P, and for which a quasi-polynomial algorithm was only recently discovered \cite{CJ+17}. 

A graph game proceeds by placing a token on a vertex in the graph, which the players move throughout the graph to produce an infinite path (``play'') $\pi$. The game is zero-sum and $\pi$ determines the winner or payoff. Two ways to classify graph games are according to the type of {\em objectives} of the players, and according to the {\em mode of moving} the token. For example, in {\em reachability games}, the objective of \PO is to reach a designated vertex $t$, and the objective of \PT is to avoid $t$. An infinite play $\pi$ is winning for \PO iff it visits $t$. The simplest mode of moving is {\em turn based}: the vertices are partitioned between the two players and whenever the token reaches a vertex that is controlled by a player, he decides how to move the token.

We study a new mode of moving in infinite-duration games, which is called {\em bidding}, and in which the players {\em bid} for the right to move the token. The bidding mode of moving was introduced in \cite{LLPSU99,LLPU96} for reachability games, where two variants of {\em first-price auctions} where studied: Each player has a budget, and before each move, the players submit sealed bids simultaneously, where a bid is legal if it does not exceed the available budget, and the higher bidder moves the token. The bidding rules differ in where the higher bidder pays his bid. In {\em Richman} bidding (named after David Richman), the higher bidder pays the lower bidder. In {\em poorman bidding}, which is the bidding rule that we focus on in this paper, the higher bidder pays the ``bank''. Thus, the bid is deducted from his budget and the money is lost. Note that while the sum of budgets is constant in Richman bidding, in poorman bidding, the sum of budgets shrinks as the game proceeds. One needs to devise a mechanism that resolves ties in biddings, and our results are not affected by the tie-breaking mechanism that is used. 

Bidding games naturally model decision-making settings in which agents need to invest resources in an ongoing manner. We argue that the modelling capabilities of poorman bidding exceed those of Richman bidding. Richman bidding is restricted to model ``scrip'' systems that use internal currency to avoid free riding and guarantee fairness. Poorman bidding, on the other hand, model a wider variety of settings since the bidders pay their bid to the auctioneer. We illustrate a specific application of infinite-duration poorman bidding in reasoning about ongoing stateful auctions, which we elaborate on in Section~\ref{sec:application}.
\begin{example}
\label{ex:auction}
Consider a setting in which two buyers compete in auction to buy $k \in \Nat$ goods that are ``rented'' for a specific time duration. For example, a webpage has $k$ ad slots, and each slot is sold for a fixed time duration, e.g., one day. At time point $1 \leq i \leq k$, good~$i$ is put up for sale in a second-price auction, where the higher bidder pays the auctioneer and keeps the good for the fixed duration of time. We focus on the first buyer. Each good entails a reward for him, and we are interested in devising a bidding strategy that maximize the long-run average of the rewards. For example, the simple case of a site with one ad slot is represented by the game that is depicted in Fig.~\ref{fig:loops}, where the vertex $v_1$ represents the case that \PO's ad appears and $v_2$ represents the case that \PT's ad appears. \PO's goal is to maximize the long-run average time that his ad appears, which intuitively amounts to ``staying'' as much time as possible in $v_1$. \PO's goal is formally described as a mean-payoff objective, which we elaborate on below. Our results on mean-payoff poorman games allow us to construct an optimal strategy for the players.\hfill\qed
\end{example}

Another advantage of poorman bidding over the Richman bidding is that their definition generalizes easily to domains in which the restriction of a fixed sum of budgets is an obstacle. For example, in ongoing auctions as described in the example above, often a good is sold to multiple buyers with partial information of the budgets. These are two orthogonal concepts that have not been studied in bidding games and are both easier to define in poorman bidding rather than in Richman bidding.



A central quantity in bidding games is the {\em ratio} of the players' initial budgets. Formally, let $B_i \in \Real_{\geq0}$, for $i \in \set{1,2}$, be Player~$i$'s initial budget. The {\em total initial budget} is $B = B_1 + B_2$  and Player~$i$'s {\em initial ratio} is $B_i/B$. The first question that arises in the context of bidding games is a necessary and sufficient initial ratio for a player to guarantee winning. For reachability games, it was shown in \cite{LLPSU99,LLPU96} that such {\em threshold ratios} exist in every  reachability  Richman and poorman game: for every vertex $v$ there is a ratio $\thresh(v) \in [0,1]$ such that (1) if \PO's initial ratio exceeds $\thresh(v)$, he can guarantee winning, and (2) if his initial ratio is less than $\thresh(v)$, \PT can guarantee winning. This is a central property of the game, which is a form of {\em determinacy}, and shows that no ties can occur.\footnote{When the initial budget of \PO is exactly $\thresh(v)$, the winner of the game depends on how we resolve draws in biddings.} 

An intriguing equivalence was observed in \cite{LLPSU99,LLPU96} between {\em random-turn games} \cite{PSSW07} and reachability bidding games, but only with Richman-bidding. For $r \in [0,1]$, the random-turn game that corresponds to a bidding game $\G$ w.r.t.\ $r$, denoted $\RT^r(\G)$, is a special case of stochastic game \cite{Con90}: rather than bidding for moving, in each round, independently, \PO is chosen to move with probability $r$ and \PT moves with the remaining probability of $1-r$. Richman reachability games are equivalent to uniform random-turn games, i.e., with $r=0.5$ (see Theorem~\ref{thm:reach} for a precise statement of the equivalence). For reachability poorman-bidding games, no such equivalence is known and it is unlikely to exist since there are (simple) finite poorman games with irrational threshold ratios. The lack of such an equivalence makes poorman games technically more complicated.

More interesting, from the synthesis and logic perspective, are infinite winning conditions, but they have only been studied in the Richman setting previously \cite{AHC19}. We show, for the first time, existence of threshold ratios in qualitative poorman games with infinite winning conditions such as parity. We show a linear reduction from poorman parity games to poorman reachability games, similarly to the proof in the Richman setting. First, we show that in a strongly-connected game, one of the players wins with any positive initial ratio, thus the {\em bottom strongly-connected components} (BSCCs, for short) of the game graph can be partitioned into ``winning'' for \PO and ``losing'' for \PO. Second, we construct a reachability poorman game in which each player tries to force the game to a BSCC that is winning for him.

Things get more interesting in {\em mean-payoff} poorman games, which are zero-sum quantitative games; an infinite play of the game is associated with a {\em payoff} which is \PO's reward  and \PT's cost, thus we respectively refer to the players in a mean-payoff game as \Max and \Min. The central question in these games is: Given a value $c \in \Q$, what is the initial ratio that is necessary and sufficient for \Max to guarantee a payoff of $c$? More formally, we say that $c$ is the {\em value} with respect to a ratio $r \in [0,1]$ if for every $\epsilon >0$, we have (1) when \Max's initial ratio is $r +\epsilon$, he can guarantee a payoff of at least $c$, and (2) intuitively, \Max cannot hope for more: if \Max's initial ratio is $r-\epsilon$, then \Min can guarantee a payoff of at most $c$.

Our most technically-involved contribution is a construction of optimal strategies in mean-payoff poorman games, which depend on the initial ratio $r \in [0,1]$.  The key component of the solution is a quantitative solution to strongly-connected games, which, similar to parity games, allows us to reduce general mean-payoff poorman games to reachability poorman games by reasoning about the BSCCs of the graph. Before describing our solution, let us highlight an interesting difference between Richman and poorman bidding. With Richman bidding, it is shown in \cite{AHC19} that a strongly-connected mean-payoff Richman-bidding game has a value that does not depend on the initial ratio and only on the structure of the game. It thus seems reasonable to guess that the initial ratio would not matter with poorman bidding as well. We show, however, that this is not the case; the higher \Max's initial ratio is, the higher the payoff he can guarantee. We demonstrate this phenomenon with the following simple game. Technically, each vertex in the graph has a weight the payoff of an infinite play $\pi$ is defined as follows. The {\em energy} of a prefix $\pi^n$ of length $n$ of $\pi$, denoted $E(\pi^n)$,  is the sum of the weights it traverses. The payoff of $\pi$ is $\lim\inf_{n\to \infty} E(\pi^n)/n$. 

\begin{example}
\label{ex:loops}
Consider the mean-payoff poorman game that is depicted in Figure~\ref{fig:loops}. We take the viewpoint of \Min in this example. We consider the case of $r=\frac{1}{2}$, and claim that the value with respect to $r=\frac{1}{2}$ is $0$. Suppose for convenience that \Min wins ties. 
Note that the players' choices upon winning a bid in the game are obvious, and the difficulty in devising a strategy is finding the right bids. Intuitively, \Min copies \Max's bidding strategy. Suppose, for example, that \Min starts with a budget of $1+\epsilon$ and \Max starts with $1$, for some $\epsilon>0$. A strategy for \Min that ensures a payoff of $0$ is based on a queue of numbers as follows: In round $i$, if the queue is empty \Min bids $\epsilon \cdot 2^{-i}$, and otherwise the maximal number in the queue. If \Min wins, he removes the minimal number from the queue (if non-empty). If \Max wins, \Min adds \Max's winning bid to the queue. For example, suppose \Max's first bid is $0.2$, he wins since \Min bids $\epsilon/2$, and \Min adds $0.2$ to the empty queue. \Min's second bid is $0.2$. Suppose \Max bids $0.3$ in the second turn, thus he wins again. \Min adds $0.3$ to the queue and bids $0.3$ in the third bidding. Suppose \Max bids $0.1$, thus \Min wins and removes $0.3$ from the queue. In the next bidding his bid is $0.2$. 

We make several observations. (1) \Min's strategy is legal: it never bids higher than the available budget. (2) The size of the queue is an upper bound on the energy; indeed, every bid in the queue corresponds to a \Max winning bid that is not ``matched'' (the size is an upper bound since \Min might win biddings when the queue is empty). (3) If \Min's queue fills, it will eventually empty. Indeed, if $b \in \Real$ is in the queue, in order to keep $b$ in the queue, \Max must bid at least $b$, thus eventually his budget runs out. Combining, since the energy is at most $0$ when the queue empties, \Min's strategy guarantees that the energy is at most $0$ infinitely often. Since we use $\lim \inf$ in the definition of the payoff,  \Min guarantees a non-positive payoff. Showing that \Max can guarantee a non-negative payoff with an initial ratio of $\frac{1}{2} + \epsilon$ is harder, and a proof for the general case can be found in Section~\ref{sec:MP}.



We show that the value $c$ decreases with \Max's initial ratio $r$. We set $r=\frac{1}{3}$. Suppose, for example, that \Min's initial budget is $2+\epsilon$ and \Max's initial budget is $1$. We claim that \Min can guarantee a payoff of $-1/3$. His strategy is similar to the one above, only that whenever \Max wins with $b$, \Min pushes $b$ to the queue twice. Observations (1-3) still hold. The difference is that now, since every \Max win is matched by two \Min wins, when the queue empties, the number of \Min wins is at least twice as much as \Max's wins, and the claim follows. 

This example shows the contrast between Richman and poorman bidding. When using Richman bidding, \Min can guarantee a payoff of $0$ with every initial budget, and cannot guarantee $-\epsilon$, even with a ratio of $1-\delta$, for any $\epsilon,\delta > 0$.\hfill\qed
\end{example}

\begin{figure}[t]
\begin{minipage}[b]{0.4\linewidth}
\centering
\includegraphics[height=1.2cm]{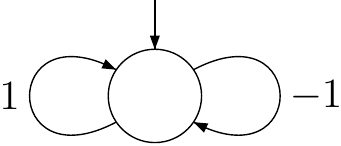}
\caption{A mean-payoff game.}
\label{fig:loops}

\end{minipage}
\begin{minipage}[b]{0.51\linewidth}
\centering
\includegraphics[height=1.2cm]{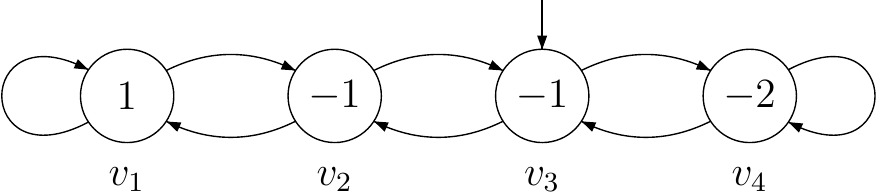}
\caption{A second mean-payoff game.}
\label{fig:ex-2}
\end{minipage}
\end{figure}

In order to solve strongly-connected mean-payoff poorman games, we identify the following equivalence with biased random-turn games. Consider a strongly-connected mean-payoff poorman game $\G$ and a ratio $r \in [0,1]$. Recall that $\RT^r(\G)$ is the random-turn game in which \Max is chosen with probability $r$ and \Min with probability $1-r$. Since $\G$ is a mean-payoff game, the game $\RT^r(\G)$ is a stochastic mean-payoff game. Its value, denoted $\MP(\RT^r(\G))$, is the optimal expected payoff that the players can guarantee, and is known to exist~\cite{MN81}. For every $\epsilon > 0$, we show that when \Max's initial ratio is $r+\epsilon$, he can guarantee a payoff of $\MP(\RT^r(\G))$, and he cannot do better: \Min can guarantee a payoff of at most $\MP(\RT^r(\G))$ with an initial ratio of $1-r+\epsilon$. Thus, the value of $\G$ w.r.t. $r$ equals $\MP(\RT^r(\G))$. One way to see this result is as a form of derandomization: we show that \Max has a deterministic bidding strategy in $\G$ that ensures a behavior that is similar to the random behavior of $\RT^r(\G)$. We find this equivalence between the two models particularly surprising due to the fact that, unlike Richman bidding, an equivalence between random-turn games and reachability poorman games is unlikely to exist. Second, while Richman games are equivalent to uniform random-turn games, we are not aware of any known equivalences between bidding games and biased random-turn games, i.e., $r\neq 0.5$.

Recall that a strongly-connected mean-payoff Richman-bidding game $\G$ has a value $c$ that does not depend on the initial ratio. The value comes from an equivalence with uniform random-turn games \cite{AHC19}: the value $c$ of $\G$ under Richman bidding equals the value of the uniform stochastic mean-payoff game $\RT^{0.5}(\G)$. That is, with Richman bidding, \Min can guarantee $c$ with an initial ratio of $\delta$, and cannot guarantee $c-\epsilon$ with an initial ratio of $1-\delta$, for every $\epsilon,\delta> 0$. One interesting corollary is that the value of $\G$ when viewed as a Richman game equals the value of $\G$ when viewed as a poorman game with respect to the initial ratio $0.5$. We are not aware of previous such connections between the two bidding rules.

Finally, we address, for the first time, complexity issues in poorman games; namely, we study the problem of finding threshold ratios in poorman games. 
We show that for qualitative games, the corresponding decision problem is in PSPACE using the existential theory of the reals \cite{Can88}. 
 For mean-payoff games, the problem of finding the value of the game with respect to a given ratio is also in PSPACE for general games, and for strongly-connected games, we show the value can be found in NP and coNP, and even in P for strongly-connected games with out-degree $2$. 


\paragraph*{Related work}
As mentioned above, bidding games can model ongoing auctions, like the ones that are used in internet companies such as Google to sell advertisement slots \cite{Mut09}. Sequential auctions, which are also ongoing, have been well studied, e.g., \cite{LST12b,Web81}, and let us specifically point \cite{GS01,Rod09}, which, similar to bidding games, studies two-player sequential auctions with perfect information. Bidding games differ from these models in two important aspects: (1) bidding games are zero-sum games, and (2) the budgets that are used for bidding do not contribute to the utility and are only used to determine which player moves. Point (2) implies that bidding games are particularly appropriate to model settings in which the budget has little or no value, similar in spirit to the well-studied {\em Colonel Blotto games} \cite{Bor21}. A dynamic version of Colonel Blotto games called {\em all-pay bidding games} has been recently studied \cite{AIT20}.  Non-zero-sum Richman-bidding games have been used to reason about ongoing negotiations \cite{MKT18}. 

Graph games are popular to reason about systems in formal methods \cite{handbookMC} and about multi-agent systems in AI \cite{AHK02}. Bidding games extend the modelling capabilities of these games and allow reasoning about multi-process systems in which a scheduler accepts payment in return for priority. {\em Blockchain} technology is one example of such a technology. Simplifying the technology, a blockchain is a log of transactions issued by clients and maintained by {\em miners}. In order to write to the log, clients send their transactions and an offer for a transaction fee to a miner, who has freedom to decide transaction priority. We expect that a more precise modelling of such systems will assist in their verification against attacks, which is a problem of special interest since bugs can result in significant losses of money (see for example,  \cite{CGV18} and a description of an attack \url{http://bit.ly/2obzyE7}). Note that poorman bidding models such settings better than Richman bidding since transaction fees are paid to the scheduler (the miners) rather than the other player. Richman bidding is appropriate when modelling ``scrip systems'' that use internal currency to prevent free-riding \cite{KFH12}, and are popular in databases for example. 

In this work, we show that mean-payoff poorman games are equivalent to biased random-turn games. Thus, there is a contrast with mean-payoff Richman games, which are equivalent to uniform random-turn games. To better understand these differences between the seemingly similar bidding rules, mean-payoff {\em taxman} games where studied in \cite{AHZ19}. Taxman bidding were defined and studied in \cite{LLPSU99} for reachability objectives span the spectrum between Richman and poorman bidding. They are parameterized by a constant $\tau \in [0,1]$: portion $\tau$ of the winning bid is paid to the other player, and portion $1-\tau$ to the bank. Thus, with $\tau = 1$ we obtain Richman bidding and with $\tau=0$, we obtain poorman bidding. It was shown that the value of a mean-payoff taxman bidding game $\G$ parameterized by $\tau$ and with initial ratio $r$ equals $\MP(\RT^{F(\tau, r)}(\G))$, for $F(\tau, r) = \frac{r+\tau\cdot (1-r)}{1+\tau}$.

To the best of our knowledge, since their introduction, poorman games have not been studied. Motivated by recreational games, e.g., bidding chess \cite{BP09,LW18}, {\em discrete bidding games} with Richman bidding rules are studied in \cite{DP10}, where the money is divided into chips, so a bid cannot be arbitrarily small unlike the bidding games we study. Infinite-duration discrete bidding games with Richman bidding and various tie-breaking mechanisms have been studied in \cite{AAH19}, where they were shown to be a largely determined sub-class of concurrent games.


\section{Preliminaries}
A graph game is played on a directed graph $G = \zug{V, E}$, where $V$ is a finite set of vertices and $E \subseteq V \times V$ is a set of edges. The {\em neighbors} of a vertex $v \in V$, denoted $N(v)$, is the set of vertices $\set{u \in V: \zug{v,u} \in E}$, and we say that $G$ has out-degree $2$ if for every $v \in V$, we have $|N(v)| = 2$. A {\em path} in $G$ is a finite or infinite sequence of vertices $v_1,v_2,\ldots$ such that for every $i \geq 1$, we have $\zug{v_i,v_{i+1}} \in E$. 

\paragraph{Objectives} 
An objective $O$ is a set of infinite paths. In reachability games, \PO has a target vertex $v_R$ and an infinite path is winning for him if it visits $v_R$. In {\em parity} games each vertex has a parity index in $\set{1,\ldots, d}$, and an infinite path is winning for \PO iff the maximal parity index that is visited infinitely often is odd. We also consider games that are played on a weighted graph $\zug{V, E, w}$, where $w: V \rightarrow \Q$. Consider an infinite path $\pi = v_1,v_2,\ldots$. For $n \in \Nat$, we use $\pi^n$ to denote the prefix of length $n$ of $\pi$. We call the sum of weights that $\pi^n$ traverses the {\em energy} of the game, denoted $E(\pi^n)$. Thus, $E(\pi^n) = \sum_{1 \leq j < n} w(v_j)$. In {\em energy games}, the goal of \PO is to keep the energy level positive, thus he wins an infinite path iff for every $n \in \Nat$, we have $E(\pi^n) > 0$. 
\stam{
\begin{remark}
\label{rem:energy}
Note that an energy game is a succinctly-represented infinite-state  reachability game. Consider an energy game $\G = \zug{V, E, w}$, where for simplicity assume $w: V \rightarrow \Nat$. We can construct a reachability game $\G'$ over the vertices $V \times (\Nat \cup \set{0})$, where we associate the position in $\G$ in which the energy is $k \in \Nat$ and the token is located on a vertex $v$, with reaching $\zug{v,k}$ in $\G'$. Recall that the goal of \PO in $\G$ is to drop the energy to $0$. Accordingly, the target of \PO in $\G'$ is reaching a vertex $\zug{v,0}$, for some $v \in V$.
\end{remark}
}
Unlike the previous objectives, a path in a {\em mean-payoff} game is associated with a payoff, which is \PO's reward  and \PT's cost. Accordingly, in mean-payoff games, we refer to \PO as \Min  and \PT as \Max. We define the payoff of $\pi$ to be $\lim\inf_{n \to \infty} \frac{1}{n} E(\pi^n)$. We say that \Max wins an infinite path of a mean-payoff game if the payoff is non-negative.

\paragraph{Strategies and plays}
A {\em strategy} prescribes to a player which {\em action} to take in a game, given a finite {\em history} of the game, where we define these two notions below. For example, in turn-based games, a strategy takes as input, the sequence of vertices that were visited so far, and it outputs the next vertex to move to. In bidding games, histories and strategies are more complicated as they maintain the information about the bids and winners of the bids. Formally, a history is a sequence $\tau = v_0,\zug{v_1,b_1, \ell_1},\zug{v_2, b_2, \ell_2}, \ldots, \zug{v_k, b_k, \ell_k} \in V \cdot (V \times \Real \times \set{1,2})^*$, where, for $j \geq 1$, in the $j$-th round, the token is placed on vertex $v_{j-1}$, the winning bid is $b_j$, and the winner is Player~$\ell_j$, and Player~$\ell_j$ moves the token to vertex $v_j$. A strategy prescribes an action $\zug{b, v}$, where $b$ is a bid that does not exceed the available budget and $v$ is a vertex to move to upon winning. The winner of the bidding is the player who bids higher, where we assume there is some mechanism to resolve draws, and our results are not affected by what the mechanism is. More formally, for $i \in \set{1,2}$, let $B_i$ be the initial budgets of \PLi, and, for a finite history $\pi$, let $W_i(\pi)$ be the sum of \PLi winning bids throughout $\pi$. In Richman bidding, the winner of a bidding pays the loser, thus \PO's budget following $\pi$ is $B_1 - W_1 + W_2$. In poorman bidding, the winner pays the ``bank'', thus \PO's budget following $\pi$ is $B_1 - W_1$. Note that in poorman bidding, the loser's budget does not change following a bidding. An initial vertex together with two strategies for the players determine a unique infinite {\em play} $\pi$ for the game. The vertices that $\pi$ visits form an infinite path $path(\pi)$. \PO wins $\pi$ according to an objective $O$ iff $path(\pi) \in O$. We call a strategy $f$ {\em winning} for \PO if for every strategy $g$ of \PT the play they determine satisfies $O$. Winning strategies for \PT are defined dually. 

\begin{definition} ({\bf Initial ratio})
Suppose the initial budget of Player~$i$ is $B_i$, for $i \in \set{1,2}$, then the {\em total initial budget} is $B = B_1 + B_2$ and Player~$i$'s {\em initial ratio} is $B_i/B$. We assume $B > 0$.
\end{definition}

The first question that arrises in the context of bidding games asks what is the necessary and sufficient initial ratio to guarantee an objective. We generalize the definition in \cite{LLPSU99,LLPU96}:

\begin{definition} ({\bf Threshold ratios})
Consider a poorman or Richman game $\G$, a vertex $v$, and an initial ratio $r$ and objective $O$ for \PO. The threshold ratio in $v$, denoted $\thresh(v)$, is a ratio in $[0,1]$ such that
\begin{itemize}
\item if $r > \thresh(v)$, then \PO has a winning strategy that guarantees $O$ is satisfied, and 
\item if $r < \thresh(v)$, then \PT has a winning strategy that violates $O$.
\end{itemize}
\end{definition}

Recall that we say that \Max wins a mean-payoff game $\G = \zug{V, E, w}$ if the mean-payoff value is non-negative. Finding $\thresh(v)$ for a vertex $v$ in $\G$ thus answers the question of what is the minimal ratio of the initial budget that guarantees winning. A more refined question asks what is the optimal payoff \Max can guarantee with an initial ratio $r$. Formally, for a constant $c \in \Q$, let $\G^c$ be the mean-payoff game that is obtained from $\G$ by decreasing all weights by $c$.

\begin{definition}
({\bf Mean-payoff values}) Consider a mean-payoff game $\G = \zug{V, E, w}$ and a ratio $r \in [0,1]$. The value of $\G$ with respect to $c$, denoted $\MP^r(\G, v)$, is such that $\thresh(v) = r$ in $\G^c$.
\end{definition}

\paragraph{Random-turn games}
In a {\em stochastic game} the vertices of the graph are partitioned between two players and a {\em nature} player. As in turn-based games, whenever the game reaches a vertex of Player~$i$, for $i =1,2$, he choses how the game proceeds, and whenever the game reaches a vertex $v$ that is controlled by nature, the next vertex is chosen according to a probability distribution that depends only on $v$. 

Consider a game $\G = \zug{V, E}$. The {\em random-turn game} with ratio $r \in [0,1]$ that is associated with $\G$ is a stochastic game that intuitively simulates the fact that \PO chooses the next move with probability $r$ and \PT chooses with probability $1-r$. Formally, we define $\RT^r(\G) = \zug{V_1, V_2, V_N, E, \Pr, w}$, where each vertex in $V$ is split into three vertices, each controlled by a different player, thus for $\alpha \in \set{1, 2, N}$, we have $V_\alpha = \set{v_\alpha: v \in V}$, nature vertices simulate the fact that \PO chooses the next move with probability $r$, thus $\Pr[v_N,v_1] = r = 1- \Pr[v_N,v_2]$, and reaching a vertex that is controlled by one of the two players means that he chooses the next move, thus $E = \set{\zug{v_\alpha, u_N}: \zug{v,u} \in E \text{ and } \alpha \in \set{1, 2}}$. When $\G$ is weighted, then the weights of $v_1, v_2$, and $v_N$ equal that of $v$. 

Fixing two strategies $f$ and $g$ for the two players in a stochastic game results in a Markov chain, which in turn gives rise to a probability distribution $D(f,g)$ over infinite sequences of vertices. A strategy $f$ is {\em optimal} w.r.t. an objective $O$ if it maximizes $\sup_f \inf_g \Pr_{\pi \sim D(f,g)}[\pi \in O]$. For the objectives we consider, it is well-known that optimal strategies exist, which are, in fact, {\em positional}; namely, strategies that only depend on the current position of the game and not on its history.

\begin{definition}
({\bf Values})
Let $r \in [0,1]$. For a qualitative game $\G$, the {\em value} of $\RT^r(\G)$, denoted $val(\RT^r(\G))$, is the probability that \PO wins when he plays optimally. For a mean-payoff game $\G$, the {\em mean-payoff value} of $\RT^r(\G)$, denoted $\MP(\RT^r(\G))$, is the maximal expected payoff \Max obtains when he plays optimally.
\end{definition}

\section{Qualitative Poorman Games}
For qualitative objectives, poorman games have mostly similar properties to the corresponding Richman games, though they are technically more complicated than Richman bidding. We start with reachability objectives, which were studied in \cite{LLPU96,LLPSU99}. The objective they study is slightly different than ours and we call it {\em double-reachability}: both players have targets and the game ends once one of the targets is reached. As we show below, for our purposes, the variants are equivalent since there are no draws in finite-state  double-reachability poorman and Richman games. 

Consider a double-reachability game $\G = \zug{V, E, u_1, u_2}$, where, for $i=1,2$, the target of Player~$i$ is $u_i$. In both Richman and poorman bidding, trivially \PO wins in $u_1$ with any initial budget and \PT wins in $u_2$ with any initial budget, thus $\thresh(u_1) = 0$ and $\thresh(u_2) = 1$. For $v \in V$, let $v^+,v^- \in N(v)$ be such that, for every $v' \in N(v)$, we have $\thresh(v^-) \leq \thresh(v') \leq \thresh(v^+)$. 

\begin{theorem}
\label{thm:reach}
\cite{LLPU96,LLPSU99} 
Threshold ratios exist in reachability Richman and poorman games. Moreover, consider a double-reachability game $\G = \zug{V, E, u_1, u_2}$. 
\begin{itemize}
\item In Richman bidding, for $v \in V \setminus \set{u_1, u_2}$, we have $\thresh(v) = \frac{1}{2}\big(\thresh(v^+)+\thresh(v^-)\big)$, and it follows that $\thresh(v) = val(\RT^{0.5}(\G,v))$ and that $\thresh(v)$ is a rational number. 
\item In poorman bidding, for $v \in V \setminus \set{u_1, u_2}$, we have $\thresh(v) = \thresh(v^+)/\big(1-\thresh(v^-)+\thresh(v^+)\big)$. There is a game $\G$ and a vertex $v$ with an irrational $\thresh(v)$.
\end{itemize}
\end{theorem}
\begin{proof}
The proof here is similar to \cite{LLPSU99} and is included for completeness, with a slight difference: unlike \cite{LLPSU99}, which assume that every vertex has a path to both targets, we also address the case where one of the targets is not reachable. This will prove helpful when reasoning about infinite-duration bidding games. The Richman case is irrelevant for us and we leave it out.

We start with the two simpler claims. Assume that in a  double-reachability poorman game $\G$, for each vertex $v$, we have $\thresh(v) = \thresh(v^+)/\big(1-\thresh(v^-)+\thresh(v^+)\big)$. We show a double-reachability poorman game with irrational threshold ratios. Consider the game with vertices $u_1, v_1, v_2,$ and $u_2$, and edges $u_1 \leftarrow v_1 \leftrightarrow v_2 \rightarrow u_2$. Solving the equation above we get $\thresh(v_1)=(\sqrt{5}-1)/2$ and $\thresh(v_2)=(3-\sqrt{5})/2$, which are irrational.

Next, we show existence of threshold ratios in a reachability poorman games by reducing them to double-reachability games. Consider a game $\G = \zug{V, E, u_1}$. Let $S \subseteq V$ be the set of vertices that have no path to $u_1$. Since \PO cannot win from any vertex in $S$, we have $\thresh(v) = 1$. Let $\G' = \zug{V', E', u_1, u_2}$ be the double-reachability game that is obtained from $\G$ by setting $V' = V \setminus S$ and \PT's target $u_2$ to be a vertex in $S$. Consider a vertex $v \in V'$. We claim that $\thresh(v)$ in $\G'$ equals $\thresh(v)$ in $\G$. Indeed, if \PO's ratio exceeds $\thresh(v)$ he can draw the game to $u_1$, and if \PT's ratio exceeds $1-\thresh(v)$ he can draw the game to $S$.

Finally, we show that every vertex in a double-reachability game has a threshold ratio. Consider a  double-reachability poorman game $\G = \zug{V, E, u_1,u_2}$. It is shown in \cite{LLPSU99} that there exists a unique function $f: V \rightarrow [0,1]$ that satisfies the following conditions: we have $f(u_1) = 0$ and $f(u_2) = 1$, and for every $v \in V$, we have $f(v) = \frac{f(v^+)}{1+f(v^+)-f(v^-)}$, where $v^+,v^- \in N(v)$ are the neighbors of $v$ that respectively maximize and minimize $f$, i.e., for every $v' \in N(v)$, we have $f(v^-) \leq f(v') \leq f(v^+)$. 

We claim that for every $v \in V$, we have $\thresh(v) = f(v)$.
Our argument will be for \PO and duality gives an argument for \PT. 
Suppose \PO's budget is $f(v)+\epsilon$ and \PT's budget is $1-f(v)$, for some $\epsilon>0$. Note that we implicitly assume that $f(v)<1$. In case $f(v)=1$ we do not show anything, but still, our dual strategy for \PT ensures that $u_2$ is visited, when the initial budget for \PT is positive. We describe a \PO strategy that forces the game to $u_1$. 

Similar to \cite{LLPSU99}, we divide \PO's budget ratio into his {\em real budget} and a {\em slush fund}. 
We will ensure the following invariants:
\begin{enumerate}
\item Whenever we are in state $v$, 
if $x$ is \PO's real budget and $y$ is \PT's budget, then $f(v)=x/(x+y)$.
\item 
Every time \PT wins a bidding the slush fund increases by a constant factor. Formally, there exists a constant $c>1$, such that when $\epsilon_0$ is the initial slush fund and $\epsilon_i$ is the slush fund after \PT wins for the $i$-th time, we have that $\epsilon_i>c \cdot \epsilon_{i-1}$, for all $i\geq 1$.
\end{enumerate}
Note that these invariants are satisfied initially.

We describe a \PO strategy. Consider a round in vertex $v$ in which \PO's real budget is $x'$, \PT's budget is $y'$ and the last time \PT won (or initially, in case \PT has not won yet) his slush fund was $\epsilon'$.  \PO's bid is $\Delta(v) \cdot x'+\delta_v \cdot \epsilon'$, where we define $\Delta(v)$ and $\delta_v$ below.
Upon winning, \PO moves to $v^{-}$, i.e., to the neighbor that minimizes $f(v)$, or, when $f(v)=0$, he moves to a vertex closer to $u_1$.
Upon winning, \PO pays $\Delta(v) \cdot x'$ from his real budget and $\delta_v \cdot \epsilon'$ from his slush fund.

For $v \in V\setminus \set{u_1, u_2}$, if $f(v) > 0$ and $f(v^-) <1$, let $\Delta(v) = \frac{f(v) - f(v^-)}{f(v)(1-f(v^-))}$ and otherwise, let $\Delta(v) = 0$. Note that the second invariant indicates that \PT cannot win more than a finite number of times, since whenever he wins, the slush fund increases by a constant and the slush fund cannot exceed $1$, because then it would be bigger than the total budget. This in turn shows that eventually \PO wins $n$ times in a row, which ensures that the play reaches $u_1$. 

We choose $\delta_v$, for $v \in V$, and show that our choice implies that \PO's strategy maintains the invariant above. Let $\Delta_{\min}$ be the smallest positive number such that $f(v)=\Delta_{\min}$ for some $v$, and $\Delta_{\min}=1$ if $f(v)=0$ for all $v \in V$.
Let $\delta_1$ be 1 and $\delta_i$ be such that $\sum_{j=1}^{i-1} \delta_j<\Delta_{\min}/2 \delta_i$, for all $i\in \{2,\dots, |V|\}$.
Also, let $\gamma$ be such that $\sum_{j=1}^{|V|} \delta_j<1/\gamma$.
For each state $v$ (such that $f(v)>0$), 
consider that \PO wins all bids and let $\dist(v)$ be the number of bids before the play ends up in $u_1$ starting from $v$. When $f(v)=0$, let $\dist(v)$ be the length of the shortest path from $v$ to $u_1$.
Then, $\delta_v=\gamma \delta_{i}$, for $i=|V|-\dist(v)$. 

In case \PO wins, his real budget becomes $x'-\Delta(v)x'$, and \PT's budget stays $y'$.
In that case, \PO's new real budget ratio becomes $\frac{(1-\Delta(v))x'}{(1-\Delta(v))x'+y'}=f(v^-)$, and the invariants are thus satisfied.
(His slush fund also decreases by $\delta_v\epsilon'$. We will not proof anything about the slush fund in this case, except noting that it stays positive).

In case \PT wins, \PO's real budget stays $x'$ and \PT's budget is at most $y'-\Delta(v)x'-\delta_v\epsilon'$.
By construction, we have that if \PT's budget became $y'-\Delta(v)x'$, then \PO's budget ratio becomes $\frac{x'}{x'+y'-\Delta(v)x'}=f(v^+)$, so even if \PT moves to $v^+$, \PT has paid $\delta_v\epsilon'$ too much for \PO's real budget ratio to be $f(v^+)$. 
Thus, the first invariant is satisfied.
Note that this also indicates that $f(v^+)\neq 1$, in this case, since otherwise \PO's budget ratio must be above 1, indicating that \PT's budget is negative. 
When $f(v^+)>0$, we can move $\delta_v\epsilon'f(v^+)/(1-f(v^+))\geq \delta_v\epsilon'\Delta_{\min}$ into the slush fund. When $f(v^+)=0$, the new slush fund is $\delta_v\epsilon'$.
Let $j$ be such that $\delta_j=\delta_v$.
By construction of $\delta_v$, we have that since the last time \PT won a bidding (or since the start if \PT never won a bid before), we have subtracted
at most $\epsilon' \sum_{i=j+1}^{|V|}\delta_i$ from the slush fond and now we have added $\delta_j\epsilon'\Delta_{\min}$. But $\delta_i$ was chosen such that $\sum_{i=j+1}^{|V|}\delta_i$ was below $\delta_v\Delta_{\min}/2$. Hence, we have added $\delta_v\epsilon'\Delta_{\min}$ to the previous content of $\epsilon'$. Because $\delta_v$ and $\Delta_{\min}$ are constants, we have thus increased the slush fund by a constant factor.
The invariants are thus satisfied in this case.
\end{proof}

We continue to study poorman games with richer objectives. 

\begin{theorem}
\label{thm:parity}
Parity poorman games are linearly reducible to  reachability poorman games. Specifically, threshold ratios exist in parity poorman games.
\end{theorem}
\begin{proof}
The crux of the proof is to show that in a bottom strongly-connected component (BSCC, for short) of $\G$, one of the players wins with every initial budget. Thus, the threshold ratios for vertices in BSCCs are either $0$ or $1$. For the rest of the vertices, we construct a reachability game in which a player's goal is to reach a BSCC that is ``winning'' for him. 

Formally, consider a strongly-connected  parity poorman game $\G = \zug{V, E, p}$. We claim that there is $\alpha \in \set{0,1}$ such that for every $v \in V$, we have $\thresh(v) = \alpha$, i.e., when $\alpha=0$, \PO wins with any positive initial budget, and similarly for $\alpha=1$. Moreover, deciding which is the case is easy: let $v_{Max} \in V$ be the vertex with maximal parity index, then $\alpha = 0$ iff $p(v_{Max})$ is odd. 

Suppose $p(v_{Max})$ is odd and the proof for an even $p(v_{Max})$ is dual. We prove in two steps. First, following the proof of  Theorem~\ref{thm:reach}, we have that when \PO's initial budget is $\epsilon >0$, he can draw the game to $v_{Max}$ once.
Second, we show that \PO can reach $v_{Max}$ infinitely often when his initial budget is $\epsilon >0$. \PO splits his budget into parts $\epsilon_1,\epsilon_2,\ldots$, where $\epsilon_i = \epsilon \cdot 2^{-i}$, for $i \geq 1$, thus $\sum_{i \geq 1} \epsilon_i = \epsilon$. Then, for $i \geq 0$, following the $i$-th visit to $v_{Max}$, he plays the strategy necessary to draw the game to $v_{Max}$ with initial budget $\epsilon_{i+1}$. 

We turn to show the reduction from  parity poorman games to double-reachability poorman games. Consider a parity poorman game $\G = \zug{V, E, p}$. Let $S \subseteq V$ be a BSCC in $\G$. We call $S$ {\em winning} for \PO if the vertex $v_{Max}$ with highest parity index in $S$ has odd $p(v_{Max})$. Dually, we call $S$ winning for \PT if $p(v_{Max})$ is even. Indeed, the claim above implies that for every $S$ that is winning for \PO and $v \in S$, we have $\thresh(v) = 0$, and dually for \PT. Let $\G'$ be a double-reachability poorman game that is obtained from $\G$ by setting the BSCCs that are winning for \PO in $\G$ to be his target in $\G'$ and the BSCCs that are winning for \PT in $\G$ to be his target in $\G'$. Similar to the proof of Theorem~\ref{thm:reach}, we have that $\thresh(v)$ in $\G$ equals $\thresh(v)$ in $\G'$, and we are done.
\end{proof}

\section{Mean-Payoff Poorman Games}
\label{sec:MP}
This section consists of our most technically challenging contribution. We construct optimal strategies for the players in  mean-payoff poorman games. The crux of the solution regards strongly-connected mean-payoff games, which we develop in the first three sub-sections. 

Consider a strongly-connected game $\G$ and an initial ratio $r \in [0,1]$. We claim that the value in $\G$ w.r.t.~$r$ does not depend on the initial vertex. For a vertex $v$ in $\G$, recall that $\MP^r(\G, v)$ is the maximal payoff \Max can guarantee when his initial ratio in $v$ is $r+\epsilon$, for every $\epsilon >0$. We claim that for every vertex $u \neq v$ in $\G$, we have $\MP^r(\G, u) = \MP^r(\G, v)$. Indeed, as in Theorem~\ref{thm:parity}, \Max can play as if his initial ratio is $\epsilon/2$ and draw the game from $u$ to $v$, and from there play using an initial ratio of $r+\epsilon/2$. Since the energy that is accumulated until reaching $v$ is constant, it does not affect the payoff of the infinite play starting from $v$.

We write $\MP^r(\G)$ to denote the value of $\G$ w.r.t.\ $r$. We show the equivalence with random-turn games: the value $\MP^r(\G)$ equals the value $\MP(\RT^r(\G))$ of the random-turn mean-payoff game $\RT^r(\G)$ in which \Max chooses the next move with probability $r$ and \Min with probability $1-r$.

\subsection{Warm up: solving a simple game}
\label{sec:warm-up}
In this section we solve a simple game through which we demonstrate the ideas of the general case. Recall that in an energy game, \Min wins a finite play if the sum of weights it traverses, a.k.a.\ the energy, is $0$ and \Max wins an infinite play in which the energy stays positive throughout the play. 

\begin{lemma}
\label{lem:loops}
\cite{LLPSU99} In the energy game that is depicted in Fig.~\ref{fig:loops}, if the initial energy is $k \in \Nat$, then \Max wins iff his initial ratio exceeds $\frac{k+2}{2k+2}$. 
\end{lemma}

The first implication in Lemma~\ref{lem:loops} is the important one for us. It shows that \Max can guarantee a payoff of $0$ with an initial budget that exceeds $0.5$. Indeed, given an initial ratio of $0.5+\epsilon$, \Max plays as if the initial energy is $k \in \Nat$ such that $\frac{k+2}{2k+2} < 0.5+\epsilon$. He thus keeps the energy bounded from below by $-k$, which implies that the payoff is non-negative.

We describe an alternative proof for the first implication in Lemma~\ref{lem:loops} whose ideas we will later generalize. We need several definitions. For $k \in \Nat$, let $S_k$ be the square of area $k^2$. In Fig.~\ref{fig:square}, we depict $S_5$. We split $S_k$ into unit-area boxes such that each of its sides contains $k$ boxes. A diagonal in $S_k$ splits it into a smaller black triangle and a larger white one. For $k \in \Nat$, we respectively denote by $t_k$ and $T_k$ the areas of the smaller black triangle and the larger white triangle of $S_k$. For example, we have $t_5 =10$ and $T_5 = 15$, and in general $t_k = \frac{k(k-1)}{2}$ and $T_k = \frac{k(k+1)}{2}$. 

\begin{figure}[ht]
\begin{minipage}{0.4\linewidth}
\center
\includegraphics[height=1.5cm]{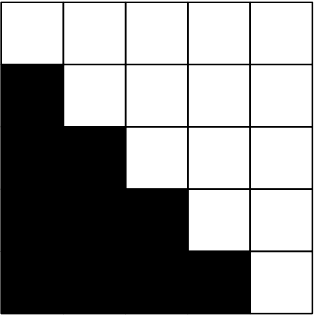}
\end{minipage}
\begin{minipage}{0.4\linewidth}
\center
\begin{tabular}{|c|c|c|c|c|c|}
\hline
& $2$ & $3$ & $4$ & $5$ & $6$ \\
\hline
 \hline
 $t_k$ & $1$ & $3$ & $6$ & $10$ & $15$ \\
 \hline
 $T_k$ & $3$ & $6$ & $10$ & $15$ & $21$ \\
 \hline
 \end{tabular}
$\ \ldots$
 \end{minipage}
\caption{The square $S_5$ with area $25$ and the sizes of some triangles.}
\label{fig:square}
\end{figure}

Suppose the game starts with energy $\kappa \in \Nat$. We show that \Max wins when his ratio exceeds $\frac{\kappa+2}{2\kappa+2}$, which equals $\frac{T_{\kappa+1}}{(\kappa+1)^2}$. For ease of presentation, it is convenient to assume that the players' ratios add up to $1+\epsilon_0$, \Max's initial ratio is $\frac{T_{\kappa+1}}{(\kappa+1)^2} + \epsilon_0$, and \Min's initial ratio is $\frac{t_{\kappa+1}}{(\kappa+1)^2}$. For $j \geq 0$, we think of $\epsilon_j$ as \Max's slush fund in the $j$-th round of the game, though its role here is somewhat less significant than in Theorem~\ref{thm:reach}. Consider a play $\pi$. We think of changes in energy throughout $\pi$ and changes in budget ratio as representing two walks on two sequences. The {\em energy sequence} is $\Nat$ and the {\em budget sequence} is $\set{t_{k}/S_{k}:k \in \Nat}$, with the natural order in the two sets. We show a strategy for \Max that maintains the invariant that whenever the energy is $k \in \Nat$, then \Max's ratio is greater than $T_{k+1}/(k+1)^2$. That is, whenever \Max wins a bidding, both sequences take a ``step up'' and when he loses, both sequences take a ``step down''.

We describe \Max's strategy. Upon winning a bidding, \Max proceeds to $v_1$, thus the energy increases by one. We assume WLog. that upon winning, \Min proceeds to $v_2$, thus the energy decreases by one. The challenge is to find the right bids. Suppose the energy level is $k$ at the $j$-th round. Thus, \Max and \Min's ratio are respectively $T_{k+1}/(k+1)^2 + \epsilon_j$ and $t_{k+1}/(k+1)^2$. In other words, \Min owns $t_{k+1}$ boxes and \Max owns a bit more than $T_{k+1}$ boxes. \Max's bid consists of two parts. \Max bids $1/(k+1)^2 + \epsilon_j/2$, or in other words,  a single box and half of his slush fund. 
We first show how the strategy maintains the invariant and then how it guarantees that an energy of $0$ is never reached. Suppose first that \Max wins the bidding. The total number of boxes decreases by one to $(k+1)^2-1$, his slush fund is cut by half, and \Min's budget is unchanged. Thus, \Max's ratio of the budget is more than $(T_{k+1}-1)/\big((k+1)^2-1\big)$, which equals $T_{k+2}/(k+2)^2$. For example, let $k=4$ and \Max's ratio exceeds  $\frac{T_5}{t_5 + T_5}$. Following a bidding win the energy increases to $k=5$ and \Max's ratio is more than $\frac{T_5-1}{t_5 + T_5 -1} = \frac{15-1}{25-1} = \frac{21}{36}=\frac{T_6}{t_6+T_6}$. In other words, we take a step up in both sequences. The other case is when \Min wins the bidding, the energy decreases by $1$, and we show that the budget sequences takes a step down. Since \Max bids more than one box, and \Min overbids, \Min bids at least one box. \Max's new ratio is more than $T_{k+1}/((k+1)^2-1) = T_k/k^2$, thus dually, both sequences take a step down. For example, again let $k=4$ and \Max's ratio exceeds $\frac{T_5}{t_5 + T_5}$. Upon losing a bidding, the energy decreases to $k=3$ and \Max's ratio is $\frac{15}{25-1} = \frac{10}{16}=\frac{T_4}{t_4+T_4}$. 

It is left to show that the energy never reaches $0$, thus the walk on the budget sequence never reaches the first element. Suppose the energy is $k=1$ in the $j$-th round, thus according to the invariant, \Max's ratio is $\frac{3}{4} + \epsilon_j$ and \Min's ratio is $\frac{1}{4}$. Recall that \Max bids $\frac{1}{(k+1)^2} + \epsilon_j/2$ at energy $k$. In particular, he bids $\frac{1}{4}+\epsilon_j/2$ at energy $1$, which exceeds \Min's budget, thus \Max necessarily wins the bidding, implying that the energy increases.

\subsection{The potential and strength of vertices}
\label{sec:potential}
In an arbitrary strongly-connected game the bids in the different vertices cannot be the same. In this section we develop a technique to determine the ``importance'' of a node $v$, which we call its {\em strength} and measures how high the bid should be in $v$ compared with the other vertices.

Consider a strongly-connected game $\G = \zug{V, E, w}$ and $r \in [0,1]$. Recall that $\RT^r(\G)$ is a random-turn game in which \Max chooses the next move with probability $r$ and \Min with probability $1-r$. A {\em positional strategy} is a strategy that always chooses the same action (edge) in a vertex. It is well known that there exist optimal positional strategies for both players in stochastic mean-payoff games. 

Consider two optimal positional strategies $f$ and $g$ in $\RT^r(\G)$, for \Min and \Max, respectively. For a vertex $v \in V$, let $v^-,v^+ \in V$ be such that $v^- = f(v_\Min)$ and $v^+ = g(v_\Max)$. The {\em potential} of $v$, denoted $\Pot^r(v)$, is a known concept in probabilistic models and its existence is guaranteed \cite{Put05}. We use the potential to define the {\em strength} of $v$, denoted $\St^r(v)$, which intuitively measures how much the potentials of the neighbors of $v$ differ. We assume w.l.o.g. that $\MP(\RT^r(\G)) = 0$ as otherwise we can decrease all weights by this value. Let $\nu \in \Q$ be such that $r = \frac{\nu}{\nu+1}$. The potential and strengths of $v$ are functions that satisfy the following:
\[
\Pot^r(v) = \frac{\nu \cdot \Pot^r(v^+) + \Pot^r(v^-)}{1+\nu}  + w(v) \text{ and }
\St^r(v) = \frac{\Pot^r(v^+) - \Pot^r(v^-)}{1+\nu}\]
There are optimal strategies for which $\Pot^r(v^-) \leq \Pot^r(v') \leq \Pot^r(v^+)$, for every $v' \in N(v)$, which can be found for example using the strategy iteration algorithm. 

Consider a finite path $\pi = v_1,\ldots, v_n$ in $\G$. We intuitively think of $\pi$ as a play, where for every $1 \leq i < n$, the bid of \Max in $v_i$ is $\St^r(v_i)$ and he moves to $v_i^+$ upon winning. Thus, if $v_{i+1} = v_i^+$, we say that \Max won in $v_i$, and if $v_{i+1} \neq v_i^+$, we say that \Max lost in $v_i$. Let $W(\pi)$ and $L(\pi)$ respectively be the indices in which \Max wins and loses in $\pi$. We call \Max wins {\em investments} and \Max loses {\em gains}, where intuitively he {\em invests} in increasing the energy and {\em gains} a higher ratio of the budget whenever the energy decreases. Let $G(\pi)$ and $I(\pi)$ be the sum of gains and investments in $\pi$, respectively, thus $G(\pi) = \sum_{i \in L(\pi)} \St^r(v_i)$ and $I(\pi) =  \sum_{i \in W(\pi)} \St^r(v_i)$. Recall that the energy of $\pi$ is $E(\pi) = \sum_{1 \leq i <n} w(v_i)$. The following lemma connects the strength, potential, and energy.

\begin{lemma}
\label{lem:magic}
Consider a strongly-connected game $\G$, a ratio $r = \frac{\nu}{1+\nu} \in (0,1)$ such that $\MP(\RT^r(\G))= 0$, and a finite path $\pi$ in $\G$ from $v$ to $u$. Then, $\Pot^r(v) - \Pot^r(u) \leq E(\pi) +  \nu \cdot G(\pi) - I(\pi)$. 
\end{lemma}
\begin{proof}
We prove by induction on the length of $\pi$. For $n=1$, the claim is trivial since both sides of the equation are $0$. Suppose the claim is true for paths of length $n$ and we prove for paths of length $n+1$. Let $\pi'$ be the prefix of $\pi$ starting from the second vertex. We distinguish between two cases. In the first case, \Max wins in $v$, thus $\pi'$ starts from $v^+$. Note that since \Max wins the first bidding, we have $G(\pi) = G(\pi')$ and $I(\pi) = \St^r(v) + I(\pi')$. Also, we have $E(\pi) = E(\pi') + w(v)$. Combining these with the induction hypothesis, we have 
\[
E(\pi) + \nu \cdot G(\pi) - I(\pi) = -\St^r(v) + w(v) + E(\pi') + \nu\cdot G(\pi') - I(\pi') \geq -\St^r(v) + w(v) + \Pot^r(v^+) - \Pot^r(u) =
\]
\[
= \frac{\Pot^r(v^-) - \Pot^r(v^+) + (1+\nu) \cdot \Pot^r(v^+)}{1+\nu} + w(v) -  \Pot^r(u) = \Pot^r(v) - \Pot^r(u)
\]

In the second case, \Max loses the first bidding, thus $\pi'$ starts from some $v'$ with $\Pot^r(v') \geq \Pot^r(v^-)$, $I(\pi) = I(\pi')$, and $G(\pi) = \G(\pi') + \St^r(v)$. We combine with the induction hypothesis to obtain the following
\[
E(\pi) + \nu \cdot G(\pi) - I(\pi) = \nu\cdot \St^r(v) + w(v) + E(\pi') + \nu\cdot G(\pi') - I(\pi') \geq \St^r(v) + w(v) + \Pot^r(v') - \Pot^r(u) \geq
\]
\[
\geq \St^r(v) + w(v) + \Pot^r(v^-) - \Pot^r(u)= \frac{\nu \cdot \Pot^r(v^+) - \nu \cdot\Pot^r(v^-) + (1+\nu) \cdot\Pot^r(v^-)}{1+\nu} + w(v) -  \Pot^r(u) =\]
\[= \Pot^r(v) - \Pot^r(u)
\]
\end{proof}

\begin{example}
Consider the game depicted in Fig.~\ref{fig:ex-2}. \Max always proceeds left and \Min always proceeds right, so, for example, we have $v^+_2 = v_1$ and $v^-_2 = v_3$. It is not hard to verify that $\MP(\RT^{2/3}(\G))=0$ by finding the stationary distribution of $\RT^{2/3}(\G)$. We have $P^\frac{2}{3}(v_1) = 6$, $P^\frac{2}{3}(v_2) = 3$, $P^\frac{2}{3}(v_3) = 0$, and $P^\frac{2}{3}(v_4) = -3$. Thus, the strengths are $\St^\frac{2}{3}(v_1) = 1$, $\St^\frac{2}{3}(v_2) = 2$, $\St^\frac{2}{3}(v_3) = 2$, and $\St^\frac{2}{3}(v_4) = 1$.  Consider the path $\pi = v_0, v_1, v_2, v_2, v_1, v_0$ in which \Max wins the first three bids and loses the last two, thus $G(\pi) = 1+2$ and $I(\pi) = 2+2+1=5$. We have $E(\pi) = -1$ since the last vertex does not contribute to the energy. The left-hand side of the expression in Lemma~\ref{lem:magic} is $0$, and the right-hand side is $-1 + 2\cdot3 -5=0$. \hfill\qed
\end{example}

\subsection{Defining a richer budget sequence}
\label{sec:walking-on-a-seq}
In this section we generalize the ideas from Section~\ref{sec:warm-up} so that we can treat any strongly-connected graph and any initial ratio. Let $r = \frac{\nu}{1+\nu}$. For the remainder of this section we fix \Min's budget to $1$ and let \Max's budget be $\nu$. We find two sequences $\set{\nu_x}_{x >0}$ and $\set{\beta_x}_{x>0}$, which we refer to as the {\em budget sequence} with properties on which we elaborate below. \Max's bid depends on the position in the budget sequence as well as the strength of the vertex. We find it more convenient to normalize the strength.
\begin{definition}
\label{def:normalized-strength}
{\bf (Normalized strength).}
Let $S = \max_v |\St^r(v)|$. The {\em normalized strength} of a vertex $v \in V$ is $\nSt(v) = \St^r(v)/S$.
\end{definition}

Formally, when the token is placed on a vertex $v \in V$ and the position of the walk is $x$, then \Max bids $\beta_x \cdot \nSt(v)$. Note that $\nSt(v) \in [0,1]$, for all $v \in V$.

We describe the intuition of the construction. We think of \Max's strategy as maintaining a position $x \in \Real_{> 0}$ on a walk, where his bidding strategy maintains the invariant that his ratio exceeds $\nu_x$. For example, in Section~\ref{sec:warm-up}, the vertices have the same importance, thus their strength is $1$. For $k \in \Nat$, we have $\nu_k = T_{k+1}/(k+1)^2$ and $\beta_k = 1/(k+1)^2$, and whenever the position is $x=k$, \Max's ratio exceeds $\nu_k$. We distinguish between two cases. Suppose first that $\nu \geq 1$. If \Max wins the bidding in $v$, then the next position of the walk is $x + \nSt(v)$, and if \Min wins the bidding, the next position is $x - \nSt(v) \cdot \nu$. When $\nu< 1$, the next position when \Max wins is $x+\nSt(v) \cdot \nu^{-1}$, and when he loses, the next position is $x-\nSt(v)$. There are two complications when comparing with the proof in Section~\ref{sec:warm-up}. 
First, while in Section~\ref{sec:warm-up}, we always take one step when winning a bidding, here the number of steps taken at a vertex $v$ depends on the importance of $v$. Unlike that proof, a step of $s \in \Q$ does not necessarily correspond to a change of $s$ in the energy. Lemma~\ref{lem:magic} guarantees that steps in the walk even out with changes in energy at the end of cycles, which suffices for our purposes. Second, that proof addresses the case of $r=1/2$ and here we consider general ratios. When \Max's initial ratio is $r$, winning a bidding is $r$-times more costly than winning a bidding for \Min. This is illustrated in Example~\ref{ex:loops}, where when \Min has a budget of $2 + \epsilon$ and \Max has a budget of $1$, \Min pushes a \Max winning bid of $b$ on the queue twice. 

 We define the following budget sequence.

\begin{definition}
\label{def:budget-seq}
Let $r = \frac{\nu}{1+\nu} > 0$ be an initial ratio. For $x>0$, we define $\nu_x = \nu(1+\frac{2}{x})$ and $\beta_x = \frac{2\cdot \min(1,\nu)}{x(x+1)}$.
\end{definition}

The most important property of the sequences is maintaining the invariant between $x$ and the ratio $\nu_x$. Recall that \Max's budget exceeds $\nu_x$ at position $x$ and \Min's budget is $1$. Suppose \Max's bid is $b$. Then, upon winning, \Max's new budget is $\nu_x - b$, and upon losing and re-normalizing \Min's budget to $1$, \Max's new budget is at least $\nu_x/(1-b)$. The following lemma shows that the invariant is maintained in both cases.

\begin{lemma}
\label{lem:invariant}
For any $0<x,\nu$ and $n\in [0,1]$, if $x(x+1)> 2\cdot n\cdot \min(1,\nu)$, we have  

\[
\frac{\nu(1+\frac{2}{x})}{1-\frac{2\cdot n\cdot \min(1,\nu)}{x(x+1)}}\geq \nu\left(1+\frac{2}{x-n\cdot \min(1,\nu)}\right) \text{ and } 
\nu(1+\frac{2}{x})-\frac{2\cdot n\cdot \min(1,\nu)}{x(x+1)}\geq \nu\left(1+\frac{2}{x+n\cdot \min(1,\nu^{-1})}\right)
\]
\end{lemma}

\begin{proof}
We start with the first claim and argue that $x(x+1)> 2\cdot n\cdot \min(1,\nu)$ implies that $x> n\min(1,\nu)$. If $x>1$, the latter follows directly from our assumptions on $n$ (and that $\min(1,\nu)\leq 1$). On the other hand, if $0<x\leq 1$, the former can be written as 
$xc>n\cdot \min(1,\nu)$, for $c=\frac{x+1}{2}\leq 1$, which in particular, implies that $x>n\cdot \min(1,\nu)$.

We have that \[
\frac{\nu(1+\frac{2}{x})}{1-\frac{2\cdot n\cdot \min(1,\nu)}{x(x+1)}}=\nu\cdot \frac{\frac{x+2}{x}}{\frac{x(x+1)-2\cdot n\cdot \min(1,\nu)}{x(x+1)}}=\nu\cdot \frac{(x+2)(x+1)}{x(x+1)-2\cdot n\cdot \min(1,\nu)} 
\]
(we have that the denominators are $>0$ since $x(x+1)> 2\cdot n\cdot \min(1,\nu)$).

Also, \[
\nu\left(1+\frac{2}{x-n\cdot \min{(1,\nu)}}\right) = \nu\left(\frac{x-n\cdot \min(1,\nu)+2}{x-n\cdot \min(1,\nu)}\right) \enspace .
\]
(we have that $x-n\cdot \min{(1,\nu)}>0$ from above).

Thus, \begin{align*}
\frac{\nu\left(1+\frac{2}{x}\right)}{1-\frac{2\cdot n\cdot \min(1,\nu)}{x(x+1)}}&\geq \nu\left(1+\frac{2}{x-n\cdot \min(1,\nu)}\right)&\Leftrightarrow\\
\frac{(x+2)(x+1)}{x(x+1)-2\cdot n\cdot \min(1,\nu)}&\geq \frac{x-n\cdot \min(1,\nu)+2}{x-n\cdot \min(1,\nu)}&\Leftrightarrow\\
(x+2)(x+1)(x-n\cdot \min(1,\nu))&\geq (x-n\cdot \min(1,\nu)+2)(x(x+1)-2\cdot n\cdot \min(1,\nu))&\Leftrightarrow\\
(x+2)(x+1)(x-n\cdot \min(1,\nu)) &- (x-n\cdot \min(1,\nu)+2)(x(x+1)-2\cdot n\cdot \min(1,\nu))\geq 0&\Leftrightarrow\\
2n\min(1,\nu)(1-n\min(1,\nu))&\geq 0&
\end{align*}

Note that $n$ and $\min(1,\nu)$ are in $[0,1]$ and thus, the inequality is true, because each factor is $\geq 0$, and we are done.

We proceed to the second claim and show that for any $0<x,\nu$ and $n\in [0,1]$, we have 
\[
\nu\left(1+\frac{2}{x}\right)-\frac{2\cdot n\cdot \min(1,\nu)}{x(x+1)}\geq \nu\left(1+\frac{2}{x+n\cdot \min(1,\nu^{-1})}\right)
\]
We have that \[
\nu\left(1+\frac{2}{x}\right)-\frac{2\cdot n\cdot \min(1,\nu)}{x(x+1)}=\nu\cdot \frac{x+2}{x}-\frac{2\cdot n\cdot \min(1,\nu)}{x(x+1)}=\frac{\nu(x+2)(x+1)-2\cdot n\cdot \min(1,\nu)}{x(x+1)} \enspace .
\]

Also, \[
\nu(1+\frac{2}{x+n\cdot \min(1,\nu^{-1})})=\nu\cdot \frac{x+n\cdot \min(1,\nu^{-1})+2}{x+n\cdot \min(1,\nu^{-1})} \enspace .
\]

Thus, \begin{align*}
\nu\left(1+\frac{2}{x}\right)-\frac{2\cdot n\cdot \min(1,\nu)}{x(x+1)}&\geq \nu\left(1+\frac{2}{x+n\cdot \min(1,\nu^{-1})}\right) &\Leftrightarrow\\
\frac{\nu(x+2)(x+1)-2\cdot n\cdot \min(1,\nu)}{x(x+1)}&\geq \nu\cdot \frac{x+n\cdot \min(1,\nu^{-1})+2}{x+n\cdot \min(1,\nu^{-1})} &\Leftrightarrow\\
(\nu(x+2)(x+1)-2\cdot n\cdot \min(1,\nu))(x+n\cdot \min(1,\nu^{-1}))&\geq \nu\cdot x(x+1)\cdot (x+n\cdot \min(1,\nu^{-1})+2)&\Leftrightarrow\\
(\nu(x+2)(x+1)-2\cdot n\cdot \min(1,\nu))(x+n\cdot \min(1,\nu^{-1}))&- \nu\cdot x(x+1)\cdot (x+n\cdot \min(1,\nu^{-1})+2)\geq 0&\Leftrightarrow\\
2n\min(1,\nu)(1-n\min(1,\nu^{-1}))&\geq 0 &
\end{align*}

Note that $n$, $\min(1,\nu)$ and $\min(1,\nu^{-1})$ are in $[0,1]$ and thus, the inequality is true, because each factor is $\geq 0$.
\end{proof}

\subsection{Putting it all together}
\label{sec:combine}
In this section we combine the ingredients developed in the previous sections to solve arbitrary strongly-connected mean-payoff games.  


\begin{theorem}
\label{thm:SCC-MP}
Consider a strongly-connected mean-payoff poorman game $\G$ and a ratio $r \in [0,1]$. The value of $\G$ with respect to $r$ equals the value of the random-turn mean-payoff game $\RT^r(\G)$ 
in which \Max chooses the next move with probability $r$, thus $\MP^r(\G) = \MP(\RT^r(\G))$. 
\end{theorem}
\begin{proof}
We assume w.l.o.g.\ that $\MP(\RT^r(\G)) = 0$ since otherwise we decrease this value from all weights. Also, the case where $r \in \set{0,1}$ is easy since $\RT^r(\G)$ is a graph and in $\G$, one of the players can win all biddings. Thus, we assume $r \in (0,1)$. Recall that $\MP(\pi) = \lim \inf_{n \to \infty} \frac{E(\pi^n)}{n}$. We show a \Max strategy that, when the game starts from a vertex $v \in V$ and with an initial ratio of $r + \epsilon$, guarantees that the energy is bounded below by a constant, which implies $\MP(\pi) \geq 0$. 

Note that showing such a strategy for \Max suffices to prove $\MP^r(\G) = 0$ since our definition for a payoff favors \Min. Consider the game $\G'$ that is obtained from $\G$ by multiplying all weights by $-1$. We associate \Min in $\G$ with \Max in $\G'$, thus an initial ratio of $1-r-\epsilon$ for \Min in $\G$ is associated with an initial ratio of $r+\epsilon$ of \Max in $\G'$. We have $\MP(\RT^{1-r}(\G'))= -\MP(\RT^{r}(\G)) = 0$. Let $f$ be a \Max strategy in $\G'$ that guarantees a non-negative payoff. Suppose \Min plays in $\G$ according to $f$ and let $\pi$ be a play when \Max plays some strategy. Since $f$ guarantees a non-negative payoff in $\G'$, we have $\lim \sup_{n \to \infty} E(\pi^n)/n \leq 0$ in $\G$, and in particular $\MP(\pi) = \lim \inf_{n \to \infty} E(\pi^n)/n \leq 0$. 

Before we describe \Max's strategy, we need several definitions. 
In Definition~\ref{def:budget-seq}, we set $\nu_x = \nu \cdot (1+2/x)$, which clearly tends to $\nu$ from above. We can thus choose $\kappa \in \Nat$ such that \Max's ratio is greater than $\nu_{\kappa}$. Suppose \Max is playing according to the strategy we describe below and \Min is playing according to some strategy. The play induces a walk on $\set{\nu_x}_{x \in \Q_{\geq 0}}$, which we refer to as the {\em budget walk}. \Max's strategy guarantees the following:

\vspace{0.05cm}
\noindent{\bf Invariant:} Whenever the budget walk reaches an $x \in \Q$, then \Max's ratio is greater than $\nu_x$. 
\vspace{0.05cm}

The walk starts in $\kappa$ and the invariant holds initially due to our choice of $\kappa$. Suppose the token is placed on the vertex $v \in V$ and the position of the walk is $x$. \Max bids $\nSt(v) \cdot \beta_x$, and he moves to $v^+$ upon winning. Suppose first that $\nu \geq 1$. If \Max wins the bidding, then the next position of the walk is $x + \nSt(v)$, and if \Min wins the bidding, the next position is $x - \nSt(v) \cdot \nu$. When $\nu< 1$, the next position when \Max wins is $x+\nSt(v) \cdot \nu^{-1}$, and when he loses, the next position is $x-\nSt(v)$. Lemma~\ref{lem:invariant} implies that in both cases the invariant is maintained.

\begin{claim}
For every \Min strategy, the budget walk stays on positive positions and never reaches $x=0$. 
\end{claim}
Suppose $\nu \geq 1$. Thus, when \Max loses with a bid of $2n/x(x+1)$, we step down $n$ steps. In order to reach $x=0$, the position needs to be $x=n$. But then, \Max's bid is $2n/n(n+1) \geq 1$, thus \Max wins the bidding since \Min's budget is $1$. Similarly, when $\nu<1$, when the bid is $2n\nu/x(x+1)$, we step down $n\cdot \nu$, and we need $x=n\cdot \nu$ to reach $x=0$. Again, since $2n\nu/n\nu(n\nu+1) \geq 1$, \Max wins the bidding.

\begin{claim}
The strategy is legal; \Max's bids never exceed his available budget.
\end{claim}
Indeed, we have $2n \min(1,\nu)/x(x+1) \leq \nu(1+2/x)$, for every $0 \leq n \leq 1$ and $\nu > 0$ since $x>0$.

\begin{claim}
The energy throughout a play is bounded from below. Formally, there exists a constant $c \in \Real$ such that for every \Min strategy and a finite play $\pi$, we have $E(\pi) \geq c$. 
\end{claim}
Consider a finite play $\pi$. We view $\pi$ as a sequence of vertices in $\G$. Recall that the budget walk starts at $\kappa$, that $G(\pi)$ and $I(\pi)$ represent sums of strength of vertices, and that $S = \max_{v \in V} |\St^r(v)|$ and $\nSt(v) = \St^r(v)/S$. Suppose the budget walk reaches $x$ following the play $\pi$. Then, when $\nu \geq 1$, we have $x = \kappa - G(\pi)/S + I(\pi)/\nu S$. Combining with $x \geq 0$, we have $S\cdot \kappa \cdot \nu \leq  -G(\pi)\cdot \nu + I(\pi)$. Let $P = \max_{u,v} \Pot^r(u) - \Pot^r(v)$. Re-writing Lemma~\ref{lem:magic}, we obtain $- G(\pi)\cdot \nu + I(\pi) \leq E(\pi) + P$. Combining the two, we have $E(\pi) \geq -P -S \cdot \kappa\cdot \nu$. Similarly, when $\nu < 1$, we have $x= \kappa - G(\pi)\cdot \nu/S + I(\pi)/S$ and combining with Lemma~\ref{lem:magic}, we obtain $E(\pi) \geq -P - S\cdot \kappa$, and we are done.
\end{proof}

\begin{remark}{\bf Richman vs poorman bidding.}
An interesting connection between poorman and Richman biddings arrises from Theorem~\ref{thm:SCC-MP}. Consider a strongly-connected mean-payoff game $\G$. For an initial ratio $r \in [0,1]$, let $\MP^r_\P(\G)$ denote the value of $\G$ with respect to $r$ with poorman bidding. With Richman bidding \cite{AHC19}, the value does not depend on the initial ratio rather it only depends on the structure of $\G$ and we can thus omit $r$ and use $\MP_\R(\G)$. Moreover, mean-payoff Richman-bidding games are equivalent to uniform random-turn games, thus $\MP_\R(\G) = \MP(\RT^{0.5}(\G))$. Our results show that poorman games with initial ratio $0.5$ coincide with Richman games. Indeed, we have $\MP_\R(\G) = \MP_\P^{0.5}(\G)$. To the best of our knowledge such a connection between the two bidding rules has not been identified before. 
\end{remark}

\begin{remark}{\bf Energy poorman games.}
The proof technique in Theorem~\ref{thm:SCC-MP} extends to  energy poorman games. Consider a strongly-connected mean-payoff game $\G$, and let $r \in [0,1]$ such that $\MP^r(\G) = 0$. Now, view $\G$ as an energy poorman game. The proof of  Theorem~\ref{thm:SCC-MP} shows that when \Max's initial ratio is $r+\epsilon$, there exists an initial energy level from which he can win the game. On the other hand, when \Max's initial ratio is $r-\epsilon$, \Min can win the energy game from every initial energy. Indeed, consider the game $\G'$ that is obtained from $\G$ by multiplying all weights by $-1$. Again, using Theorem~\ref{thm:SCC-MP} and associating \Min with \Max, \Min can keep the energy level bounded from above, which allows him, similar to the qualitative case, to play a strategy in which he either wins or increases his ratio by a constant. Eventually, his ratio is high enough to win arbitrarily many times in a row and drop the energy as low as required.
\end{remark}

\begin{remark}{\bf A general budget sequence.}
The proof of Theorem~\ref{thm:SCC-MP} uses four properties of the ``budget sequence'' $\set{\nu_x}_{x \geq 0}$ and $\set{\beta_x}_{x \geq 0}$ that is defined in Definition~\ref{def:budget-seq}: (1) the invariant between \Max's ratio and $r_x$ is maintained (shown in Lemma~\ref{lem:invariant}), (2) the bids never exceed the available budget, (3) $\lim_{x \to \infty} \nu_x = \nu$, and (4) the walk never reaches $x=0$. The existence of a budget sequence with these properties is shown in~\cite{AHZ19} for {\em taxman bidding}, which generalize both Richman and poorman bidding: taxman bidding is parameterized with a constant $\tau \in [0,1]$, where the higher bidder pays portion $\tau$ of his bid to the other player and portion $(1-\tau)$ to the bank. Unlike that proof, we define an explicit budget sequence for poorman bidding. 
\end{remark}

\subsection{Extention to general mean-payoff games}
\label{sec:general-MP}
We extend the solution in the previous sections to general graphs in a similar manner to the qualitative case; we first reason about the BSCCs of the graph and then construct an appropriate reachability game on the rest of the vertices. Recall that, for a vertex $v$ in a mean-payoff game, the ratio $\thresh(v)$ is a necessary and sufficient initial ratio to guarantee a payoff of $0$. 

Consider a mean-payoff poorman game $\G = \zug{V, E, w}$. Recall that, for $v \in V$, $\thresh(v)$ is the necessary and sufficient initial ratio for \Max to guarantee a non-positive payoff. Let $S_1,\ldots,S_k \subseteq V$ be the BSCCs of $\G$ and $S = \bigcup_{1 \leq i \leq k}S_i$. For $1 \leq i \leq k$, the mean-payoff poorman game $\G_i = \zug{S_i, E_{|S_i}, w_{|S_i}}$ is a strongly-connected game. We define $r_i \in [0,1]$ as follows. If there is an $r \in [0,1]$ such that $\MP^r(\G_i) = 0$, then $r_i =r$. Otherwise, if for every $r$, we have $\MP^r(\G_i) > 0$, then $r_i=0$, and if for every $r$, we have $\MP^r(\G_i) < 0$, then $r_i =1$. By Theorem~\ref{thm:SCC-MP}, for every $v \in S_i$, we have $\thresh(v) = r_i$. We construct a {\em generalized reachability game} $\G'$ that corresponds to $\G$ by replacing every $S_i$ in $\G$ with a vertex $u_i$. \PO wins a path in $\G$ iff it visits some $u_i$ and when it visits $u_i$, \PO's ratio is at least $r_i$. It is not hard to generalize the proof of Theorem~\ref{thm:reach} to generalized reachability poorman games and obtain  the following.

\stam{
Recall that $\thresh(v)$ is a necessary and sufficient ratio for \Max to guarantee a non-negative payoff. We characterize the threshold ratios in $\G$. Let $S_1,\ldots,S_k \subseteq V$ be the BSCCs of $\G$ and $S = \bigcup_{1 \leq i \leq k}S_i$. For $1 \leq i \leq k$, the poorman mean-payoff game $\G_i = \zug{S_i, E_{|S_i}, w_{|S_i}}$ is a strongly-connected game. Let $r_i \in \Real$ be such that $\MP^{r_i}(\G_i)=0$. By Theorem~\ref{thm:SCC-MP}, for every $v \in S_i$, we have $\thresh(v) = r_i$. We extend to the vertices that are not in $S$ using a similar idea to poorman reachability games. Let $g:V \rightarrow [0,1]$ such that for (1) every $1 \leq i \leq k$ and $v \in S_i$, we have $g(v) = r_i$, (2) for every $v \in (V \setminus S)$, there are two vertices $v^-$ and $v^+$ such that $g(v) = \frac{g(v^+)}{1-g(v^-)+g(v^+)}$, and 
(3) for every $v' \in N(v)$, we have $g(v^-) \leq g(v') \leq g(v^+)$. It is not hard to adapt the proof of Theorem~\ref{thm:reach} to show a \Max strategy that guarantees the following: When the game starts in a vertex $v \in V$ and \Max's initial ratio is $g(v) + \epsilon$, then he forces the game to reach some BSCC $S_i$, for some $1 \leq i \leq k$, such that when reaching $S_i$, \Max's ratio exceeds $r_i$. From there, he uses the strategy described in Theorem~\ref{thm:SCC-MP}. 
}

\begin{theorem}
\label{thm:MP-general}
The threshold ratios in a mean-payoff poorman game $\G$ coincide with the threshold ratios in the generalized reachability game that corresponds to $\G$. 
\end{theorem}

\subsection{Applying bidding games in reasoning about auctions for online advertisements}
\label{sec:application}
In this section we show an application of mean-payoff poorman-bidding games in reasoning about auctions for online advertisements. A typical webpage has {\em ad slots}; e.g., in Google's search-results page, ads typically appear above or beside the ``actual'' search results. Different slots have different value depending on their positions; e.g., slots at the top of the page are typically seen first, thus generate more clicks and are more valuable. A large chunk of the revenue of companies like Google comes from auctions for allocating ad slots that they regularly hold between advertisement companies. 

Consider the following auction mechanism. At each time point (e.g., each day), a slot is auctioned and the winner places an ad in the slot. It is common practice in auctions for online ads to hold {\em second-price} auctions; namely, the higher bidder sets the ad and pays the bid of the second-highest bidder to the auctioneer. Suppose there are $k \in \Nat$ ad slots. We take the view-point of an advertiser. The state of the webpage is given by $\bar{s} \in \set{0,1}^k$, where an advertiser's ad appears in a slot $1 \leq i \leq k$ iff $s_i = 1$. We assume that we are given a reward function $\rho: \set{0,1}^k \rightarrow \Q$ that assigns the utility obtained from each state $\bar{s} \in \set{0,1}^k$; e.g., the reward can be the expected revenue, which is the expected number of clicks on his ads times the expected revenue from each click. The utility for an infinite sequence $\bar{s_1},\bar{s_2},\ldots$ is the mean-payoff of $\rho(\bar{s_1}),\rho(\bar{s_2}),\ldots$. We are interested in finding an optimal bidding strategy in the ongoing auction under two simplifying assumptions: (1) the utility is obtained only from the ads and does not include the price paid for them, and (2) we assume two competitors and full information of the budgets. We obtain an optimal bidding strategy by finding an optimal strategy for \Max in a mean-payoff poorman-bidding game. In Section~\ref{sec:disc}, we discuss extensions of the bidding games that we study in this paper, that are needed to weaken the two assumptions above.

As a simple example, the special case of one ad slot is modelled as the game in Fig.~\ref{fig:loops}: in each turn the ad slot is auctioned, \Max gets a reward of $1$ when his ad shows and a penalty of $-1$ when the competitor's ad is shown. We formalize the general case. Consider an ongoing auction with $k$ slots and a reward function $\rho$. We construct a mean-payoff poorman-bidding game $\A_{k,\rho} = \zug{V, E, w}$ as follows. We define $V = \set{1,\ldots,k} \times \set{0,1}^k$. Consider $v = \zug{\ell, \bar{s}} \in V$, where $1 \leq \ell \leq k$ and $\bar{s} = \zug{s_1,\ldots, s_k} \in \set{0,1}^k$. The vector $\bar{s}$ represents the state of the webpage following the previous bidding. The slot that is auctioned at $v$ is $\ell$, thus the vertex $v$ has two neighbors $u_1 = \zug{\ell^1,\bar{s^1}}$ and $u_2=\zug{\ell^2,\bar{s^2}}$ with $\ell^1 = \ell^2 = \ell+1 \mod k$.  The state of the slots apart from the $\ell$-th slot stay the same, thus for every $i \neq \ell$, we have $s^1_i = s^2_i = s_i$. The vertex $u_1$ represents a \Max win in the bidding and $u_2$ a \Max lose, thus $s^1_\ell = 1$ and $s^2_\ell = 0$. Finally, the weight  of $v$ is $\rho(\bar{s})$. Note that $\A_{k,\rho}$ is a strongly-connected mean-payoff poorman-bidding game.

\begin{theorem}
\label{thm:app}
Consider a second-price ongoing auction with $k$ slots and a reward function $\rho$. An optimal strategy for \Max in the poorman-bidding game $\A_{k,\rho}$ coincides with an optimal bidding strategy in the auction.
\end{theorem}
\begin{proof}
The only point that requires proof is that mean-payoff poorman-bidding games are equivalent to mean-payoff games with second-price auctions. Consider a strongly-connected mean-payoff game $\G$. Let $r \in (0,1)$. Suppose \Max's initial budget is $r+\epsilon$, for $\epsilon > 0$. Theorem~\ref{thm:SCC-MP} constructs a \Max strategy $f$ that guarantees a payoff of at least $\MP(\RT^r(\G))$ under poorman bidding rules. A close look at this strategy reveals that it ensures a payoff of at least $\MP(\RT^r(\G))$ under second-price rules. Indeed, let $b$ be the \Max bid prescribed by $f$ following a finite play. Then, if \Max wins the bidding, his payment is at most $b$. On the other hand, if \Min wins the bidding, he pays at least $b$. In both cases the invariant on \Max's budget is maintained as in the proof of Theorem~\ref{thm:SCC-MP}. Finally, a dual argument as in Theorem~\ref{thm:SCC-MP} shows that \Min can guarantee a payoff of at most $\MP(\RT^r(\G))$ with second-price bidding rules. We thus conclude that the value of $\G$ under second-price bidding coincides with the value under poorman bidding, and we are done.
\end{proof}

We can use Theorem~\ref{thm:app} to answer questions of the form ``can an advertiser guarantee that his ad shows at least half the time, in the long run?''. Indeed, set $\rho(\bar{s}) = 1$ when the ad shows and $\rho(\bar{s}) = 0$ when it does not. Then, the payoff corresponds to the long-run average time that the ad shows.

\section{Computational Complexity}
We study the complexity of finding the threshold ratios in poorman games. We formalize this search problem as the following decision problem. Recall that threshold ratios in reachability poorman games may be irrational (see Theorem~\ref{thm:reach}).

\begin{center}
{\bf \THRESHBUDG} Given a bidding game $\G$, a vertex $v$, and a ratio $r \in [0,1] \cap \Q$, decide whether $\thresh(v) \geq r$. 
\end{center}

\stam{
We study the complexity of finding the threshold ratios in poorman games. We formalize this search problem in the two following problems. Recall that threshold ratios in poorman reachability games may be irrational (see Theorem~\ref{thm:reach}). 
\begin{itemize}
\item {\bf \THRESHBUDG} Given a bidding game $\G$, a vertex $v$, and a ratio $r \in [0,1]$, decide whether $\thresh(v) \geq r$. 
\item {\bf \ATHRESHBUDG} Given a bidding game $\G$, a vertex $v$, a ratio $r \in ([0,1] \cap \Q)$, and $\epsilon \in \Q_{>0}$, assuming that $\thresh(v)\not\in (r-\epsilon,r+\epsilon)$, decide whether $\thresh(v) \geq r + \epsilon$ or $\thresh(v) \leq r - \epsilon$.
\end{itemize}
Note that given a game $\G$, a vertex $v$, and $\epsilon >0$, in order to find a ratio $r \in [0,1]$ that is within $\epsilon$ of $\thresh(v)$, we can use an algorithm for the \ATHRESHBUDG problem as a subroutine in binary search. 
}

\begin{theorem}
\label{thm:qual-compl}
For poorman parity games, \THRESHBUDG is in PSPACE.
\end{theorem}
\begin{proof}
To show membership in PSPACE, we guess the optimal moves for the two players. To verify the guess, we construct a program of the  {\em existential theory of the reals} that uses the relation between the threshold ratios that is described in Theorem~\ref{thm:reach}.  Deciding whether such a program has a solution is known to be in PSPACE \cite{Can88}. 
Formally, given a parity poorman game $\G = \zug{V, E, p}$ and a vertex $v \in V$, we guess, for each vertex $u \in V$, two neighbors $u^+, u^- \in N(u)$.  We construct the following program. For every vertex $u \in V$, we introduce a variable $x_u$, and  we add constraints so that a satisfying assignment to $x_u$ coincides with the threshold ratio in $u$. Consider a BSCC $S$ of $\G$. Recall that the threshold ratios in $S$ are all either $0$ or $1$, and verifying which is the case can be done in linear time. Suppose the threshold ratios are $\alpha \in \set{0,1}$. We add constraints $x_u = \alpha$, for every $u \in S$. For every vertex $u \in V$ that is not in a BSCC, we have constraints $x_u = \frac{x_{u^+}}{1-x_{u^-}+x_{u^+}}$ and $x_{u^-} \leq x_{u'} \leq x_{u^+}$, for every $u' \in N(u)$. By Theorems~\ref{thm:reach} and~\ref{thm:parity}, a satisfying assignment assigns to $x_u$ the ratio $\thresh(u)$. We conclude by adding a final constraint $x_v \geq r$. Clearly, the program has a satisfying assignment iff $\thresh(v) \geq r$, and we are done.
\end{proof}

We continue to study mean-payoff games.  

\begin{theorem}
\label{thm:MP-complex}
For mean-payoff poorman games, \THRESHBUDG is in PSPACE. For strongly-connected 
games, it is in NP and coNP. For strongly-connected games with out-degree $2$, \THRESHBUDG is  in P. 
\end{theorem}
\begin{proof}
To show membership in PSPACE, we proceed similarly to the qualitative case, and show a nondeterministic polynomial-space that uses the existential theory of the reals to verify its guess. Given a game $\G$, we construct a program that finds, for each BSCC $S$ of $\G$, the threshold ratio for all the vertices in $V$. We then extend the program to propagate the threshold ratios to the rest of the vertices, similar to Theorem~\ref{thm:MP-general}. Given a strongly-connected game $\G$ and a ratio $r\in [0,1]$, we construct $\RT^{r}(\G)$ in linear time. Then, deciding whether $\MP(\RT^{r}(\G)) \geq 0$, is known to be in NP and coNP.

The more challenging case is the solution for strongly-connected games with out-degree $2$. Consider such a game $\G = \zug{V, E, w}$ and $r \in [0,1]$. We construct an MDP $\D$ on the structure of $\G$ such that $\MP(\D) = \MP^r(\G)$. Since finding $\MP(\D)$ is known to be in P, the claim follows. When $r \geq \frac{1}{2}$, then $\D$ is a max-MDP, and when $r < \frac{1}{2}$, it is a min-MDP. Assume the first case, and the second case is similar. We split every vertex $v \in V$ in three, where $v \in V_\Max$ and $v_1, v_2 \in V_N$. Suppose $\set{u_1, u_2} = N(v)$. Intuitively, moving to $v_1$ means that \Max prefers moving to $u_1$ over $u_2$. Thus, we have $\Pr[v_1, u_1] = r = 1-\Pr[v_1, u_2]$ and $\Pr[v_2,u_1] = 1-r = 1-\Pr[v_2, u_2]$. It is not hard to see that $\MP(\D) = \MP^r(\G)$. 
\end{proof}

\section{Discussion}
\label{sec:disc}
We studied for the first time infinite-duration poorman-bidding games. Historically, poorman bidding has been studied less than Richman bidding, but the reason was technical difficulty, not lack of motivation. In practice, while the canonical use of Richman bidding is a richer notion of fairness, poorman bidding, on the other hand, are more common since they model an ongoing investment from a bounded budget. We show the existence of threshold ratios for poorman games with qualitative objectives. For mean-payoff poorman games, we construct optimal strategies with respect to the initial ratio of the budgets. We show an equivalence between mean-payoff poorman games and random-turn games, which, to the best of our knowledge, is the first such equivalence for poorman bidding. Unlike Richman bidding for which an equivalence with random-turn games holds for reachability objectives, for poorman bidding no such equivalence is known. We thus find the equivalence we show here to be particularly surprising.

We expect the mathematical structure that we find for poorman bidding to be useful in adding to these games concepts that are important for modelling practical settings. For example, our modelling of ongoing auctions made two simplifying assumptions: (1) utility is only obtained from the weights in the graph, and (2) two companies compete for ads and there is full information on the company's budgets. Relaxing both assumptions are an interesting direction for future work. Relaxing the second assumption requires an addition of two orthogonal concepts that were never studied in bidding games: multiple players and partial information regarding the budgets. Finally, the deterministic nature of bidding games is questionable for practical applications, and a study of probabilistic behavior is initiated in \cite{AHIN19}. 

To the best of our knowledge, we show the first complexity upper bounds on finding threshold ratios in poorman games. We leave open the problem of improving the bounds we show; either improving the PSPACE upper bounds or showing non-trivial lower bounds, e.g., showing ETR-hardness. Since threshold ratios can be irrational, we conjecture that the problem is at least {\em Sum-of-squares}-hard. The complexity of finding threshold ratios in un-directed reachability Richman-bidding games (a.k.a. ``tug-of-war'' games) was shown to be in P in \cite{LLPSU99}, thereby solving the problem for uniform undirected random-turn games. Recently, the solution was extended to un-directed biased reachability random-turn games  \cite{PS19}.

This work belongs to a line of works that transfer concepts and ideas between three areas with different takes on game theory: formal methods, algorithmic game theory \cite{NRTV07}, and AI.
Examples of works in the intersection of these fields include logics for specifying multi-agent systems \cite{AHK02,CHP10,MMPV14}, studies of equilibria in games related to synthesis and repair problems \cite{CHJ06,Cha06,FKL10,AAK15}, non-zero-sum games in formal verification \cite{CMJ04,BBPG12}, and applying concepts from formal methods to {\em resource allocation games}; e.g., network games with rich specifications \cite{AKT16} and an efficient reasoning about very large games \cite{AGK17,KT17}.

\small
\bibliographystyle{plain}
\bibliography{../ga}

\end{document}